\documentclass[runningheads,envcountsame]{llncs}
\usepackage[T1]{fontenc}
\RequirePackage{amsmath,amsfonts,amssymb,todonotes}

\usepackage{hyperref} %
\usepackage[capitalize,noabbrev]{cleveref}
\usepackage{pgfplots}
\usepackage{tikz}

\usetikzlibrary{arrows,decorations,intersections,pgfplots.fillbetween}
\usetikzlibrary{shapes.multipart,shapes.misc}
\usepgflibrary{decorations.pathreplacing}
\usetikzlibrary{arrows,topaths}
\usetikzlibrary{automata,shapes.arrows, backgrounds, fadings,intersections, positioning}
\usetikzlibrary{shadows}

\usetikzlibrary{decorations.pathmorphing}
\usetikzlibrary{decorations.shapes}
\usetikzlibrary{shapes}
\usetikzlibrary{backgrounds}
\usetikzlibrary{fit}
\usetikzlibrary{arrows}
\usetikzlibrary{calc}
\usetikzlibrary{mindmap}

\tikzset{every picture/.style={thick,>=angle 60}}
\tikzset{Grand/.style={draw,circle,minimum size=11*1.5,inner sep=0}}
\tikzset{Gmax/.style={draw,rectangle,minimum size=9*1.5,inner sep=0}}
\tikzset{Gmin/.style={draw,diamond,minimum size=9*1.5,inner sep=0}}
\tikzset{gamebad/.style={fill=red}}

\definecolor{myorange}{RGB}{255, 128, 0}

\usepackage{multirow}
\usepackage{color}

\newcommand{\N}{\mathbb N}

\newcommand{\eps}{\varepsilon}

\newcommand{\eqby}[2][=]{\stackrel{\text{{\tiny{#2}}}}{#1}}
\newcommand{\eqdef}{\eqby{def}}

\newcommand{\reachset}{T}
\newcommand{\dist}{\mathcal{D}}

\newcommand{\reach}[1]{\mathtt{Reach}(#1)}
\newcommand{\reachn}[2]{\mathtt{Reach}_{#1}(#2)}

\newcommand{\avoid}[1]{\mathtt{Avoid}(#1)}

\newcommand{\game}{{\mathcal G}}

\newcommand{\states}{S}
\newcommand{\state}{s}
\newcommand{\playset}{E}
\newcommand{\zstrat}{\sigma}
\newcommand{\ostrat}{\pi}

\newcommand{\zstratset}{\Sigma}
\newcommand{\ostratset}{\Pi}

\newcommand{\expectval}{{\mathcal E}}
\newcommand{\probm}{{\mathcal P}}

\newcommand{\valueo}[1]{{\mathtt{val}_{#1}}}
\newcommand{\valueof}[2]{{\mathtt{val}_{#1}(#2)}}
\newcommand{\valueoflower}[2]{{\mathtt{val}^\downarrow_{#1}(#2)}}
\newcommand{\valueofupper}[2]{{\mathtt{val}^\uparrow_{#1}(#2)}}

\newcommand{\sca}[2]{\langle#1,#2\rangle}
\DeclareMathOperator*{\argmin}{arg\,min}

\begin{document}
\title{Memoryless Strategies in \\ Stochastic Reachability Games}
%
%
\author{Stefan Kiefer\inst{1} \and
Richard Mayr\inst{2} \and
Mahsa Shirmohammadi\inst{3} \and
Patrick Totzke\inst{4}
}
\authorrunning{S. Kiefer et al.}
%
\institute{University of Oxford, UK \and
University of Edinburgh, UK \and
IRIF \& CNRS, Universit\'e Paris cit\'e, France \and
University of Liverpool, UK
}
\maketitle              
\begin{abstract}
We study concurrent stochastic reachability games played on finite graphs.
Two players, Max and Min, seek respectively to maximize and minimize the
probability of reaching a set of target states.
We prove that Max has a memoryless strategy
that is optimal from all states that have an optimal strategy.
Our construction provides an alternative proof of this result by Bordais,
Bouyer and Le Roux~\cite{BordaisB022}, and strengthens it, as we allow
Max's action sets to be countably infinite.

\end{abstract}

\section{Introduction}

\paragraph{Background.}
We study 2-player zero-sum stochastic games.
First introduced by Shapley in his seminal 1953 work~\cite{shapley1953},
they model dynamic interactions in which the environment
responds randomly to players' actions.
Shapley's games were generalized in \cite{Gillette1958} and \cite{KumarShiau}
to allow infinite state and action sets and non-termination.

In \emph{concurrent games}, the two players (Max and Min) jointly create an infinite path through a directed graph. In each round of the play, both independently choose
an action. The next state is then determined according to a pre-defined distribution that depends on the current state and the chosen pair of actions.
\emph{Turn-based games} (also called switching-control games) are a subclass where
each state is owned by some player and only this player gets to choose an action.
These have received much attention by computer
scientists, e.g., \cite{GimbertH10,chen2013prism,bouyer2016,KieferMSW17a,BertrandGG17}.
An even more special case of stochastic games are \emph{Markov Decision Processes (MDPs)}
where all states are owned by Max (i.e., Min is passive).
MDPs are also called \emph{games against nature}.

We consider \emph{reachability objectives} which are defined w.r.t.\ a given subset of \emph{target} states.
A play is defined as winning for Max iff it
visits a target state at least once.
Thus Max aims to maximize the probability that the target set is reached.
Dually, Min aims to minimize the probability of reaching the target.
So Min pursues the dual \emph{safety objective} of
avoiding the target.

Reachability is arguably the simplest objective in games on
graphs. It can trivially be encoded into reward-based objectives; i.e.,
every play that reaches the target gets reward~$1$ and all other plays get
reward~$0$.

In turn-based reachability games over finite state spaces,
there always exist optimal Max strategies that are memoryless (choices only depend on the current state and not the history of the play; this is also
called positional) and deterministic (always chose on action as opposed to randomising among several) \cite{CONDON1992203}, \cite[Proposition 5.6.c, Proposition 5.7.c]{kucera_2011}.

This does not carry over to finite concurrent reachability games.
E.g., in the \emph{snowball game} (aka \emph{Hide-or-Run} game) \cite[Example 1]{Everett1957}\cite{KumarShiau,AlfaroHK98}
(also see \Cref{example:snowball} later in this paper),
Max does not have an optimal strategy.
However, by \cite[Corollary~3.9]{Secchi97}, Max always has $\eps$-optimal randomized
memoryless  strategies, and this holds even in countably infinite
reachability games with finite action sets.
(Deterministic strategies are generally useless in concurrent games, even in
very simple games such as \emph{Rock-Paper-Scissors}.)

Even though optimal Max strategies do not always exist, it is still
interesting how much memory they need in those instances where they do exist.
For finite concurrent reachability games with finite action sets,
it was recently shown by Bordais, Bouyer and Le Roux~\cite{BordaisB022} that optimal Max strategies,
if they exist, can be chosen as randomized memoryless.
The proof is constructive and iteratively builds the strategy on the finite
state space. To show the correctness, one needs to argue about the performance
of the constructed strategy. The proof in \cite{BordaisB022}
heavily relies on the properties of induced finite MDPs,
obtained by fixing one strategy in the game. In particular, it uses the
existence of end components in these finite MDPs and their particular
properties.

\paragraph{Our contribution.}
We give an alternative proof of this result by Bordais,
Bouyer and Le Roux~\cite{BordaisB022}.
While our proof is also constructive, it is simpler and works
directly from first principles on games, without using properties
of induced MDPs.
Moreover, it uses a construction that we call ``leaky games'', by which we
reduce the reachability objective to its dual safety objective.
Finally, our result is slightly stronger, because it holds even if
Max is allowed countably infinite action sets (while Min still has
finite action sets).
\footnote{It may be possible to generalize the proof in \cite{BordaisB022}
to countably infinite Max actions sets, but this would require (at least)
a generalization of the fixpoint theorem \cite[Theorem 12]{BordaisB022} to this setting.}


This result requires the state space to be finite.
However, for other results in this paper we allow the state space and action sets to be countably infinite, unless explicitly stated otherwise.

\section{Preliminaries}

A \emph{probability distribution} over a countable set $\states$ is a function
$\alpha:\states\to[0,1]$ with $\sum_{\state\in \states}\alpha(\state)=1$.
The \emph{support} of~$\alpha$ is the set $\{\state \in \states \mid \alpha(\state) >0\}$.
We write $\dist(\states)$ for the set of all probability distributions over~$\states$.
For $\alpha \in \dist(\states)$ and a function~$v : S \to \mathbb{R}$ we write $\sca{\alpha}{v} \eqdef \sum_{s \in S} \alpha(s) v(s)$ for the expectation of~$v$ with respect to~$\alpha$.

\subsection*{Stochastic Games}

We study perfect information stochastic games between two players, Max(imizer) and Min(imizer).
A \emph{(concurrent) game}~$\game$ is played on a countable set of states~$\states$.
For each state $\state \in \states$ there are nonempty countable \emph{action} sets
$A(\state)$ and $B(\state)$ for Max and Min, respectively.
A \emph{mixed action} for Max (resp.\ Min) in state $\state$
is a distribution over $A(\state)$ (resp.\ $B(\state)$).

Let $Z \eqdef \{(\state,a,b) \mid \state \in \states,\ a \in A(\state),\ b \in B(\state)\}$.
For every triple $(\state,a,b) \in Z$
there is a distribution $p(\state,a,b) \in \dist(\states)$ over successor states.
We call a state~$s \in S$ a \emph{sink} state if $p(s,a,b) = s$ for all $a \in A(s)$ and $b \in B(s)$.
We extend the \emph{transition function}~$p$ to mixed actions $\alpha \in \dist(A(s))$ and $\beta \in \dist(B(s))$ by letting
\[p(s,\alpha,\beta)\,\,\eqdef \sum_{a \in A(s)} \sum_{b \in B(s)} \alpha(a) \beta(b) p(s,a,b),\] which is a distribution over~$S$.
A \emph{play} from an initial state $\state_0$ is an infinite
sequence in $Z^\omega$ where the first triple contains $\state_0$.
Starting from~$\state_0$, the game is played in stages $\N=\{0,1,2,\dots\}$.
At every stage $n \in \N$, the play is in some state~$\state_n$.
Max chooses a mixed action $\alpha_n \in \dist(A(\state_n))$ and Min chooses
a mixed action~$\beta_n \in \dist(B(\state_n))$.
The next state $\state_{n+1}$ is then chosen according to the distribution
$p(\state_n,a_n,b_n)$.

\subsection*{Strategies and Probability Measures}
The set of \emph{histories} at stage $n$, with $n\in \mathbb{N}$, is denoted by $H_n$. That is,
$H_0 \eqdef \states$ and $H_n \eqdef Z^n \times \states$ for all $n>0$.
Let $H \eqdef \bigcup_{n \in \N} H_n$ be the set of all histories; note that~$H$ is countable.
For each history $h = (s_0,a_0,b_0) \cdots (s_{n-1},a_{n-1},b_{n-1}) s_{n} \in H_n$, let $\state_h \eqdef s_{n}$ denote the final state in $h$.

A \emph{strategy} for Max
is a function $\zstrat$ that to each history~$h \in H$ assigns
a mixed action $\zstrat(h) \in \dist(A(\state_h))$.
Denote by $\zstratset$ the set of strategies for
Max.
Analogously, a \emph{strategy} for  Min
is a function  $\ostrat$ that to each history~$h$ assigns
a mixed action \ $\ostrat(h) \in \dist(B(\state_h))$, and  $\ostratset$ denotes the set of strategies for Min.
A Max strategy is called \emph{memoryless} if $\sigma(h)$ depends only on~$s_h$; i.e., for all $h,h' \in H$ with $s_h = s_{h'}$ we have $\sigma(h) = \sigma(h')$.
A memoryless strategy~$\sigma$ is fully determined by~$(\sigma(s))_{s \in S}$.
A \emph{memory-based strategy} bases its decisions not only on the current
state, but also on the current mode of its memory, and it can update its
memory at every step, depending on the observed events in this step.
Memory is called \emph{public} if the content is also observable by the opposing
player and \emph{private} otherwise.
Finite-memory strategies use a memory with only finitely many different
possible modes.
A step counter is a special case of infinite memory in the form of a discrete
clock that gets incremented at every step, independently of the actions of the
players. Strategies that use just a step counter are also called
\emph{Markov strategies}.

An initial state $\state_0$ and a pair of strategies $\zstrat, \ostrat$
for Max and Min
induce a probability measure on sets of plays.
We write $\probm_{\game,\state_0,\zstrat,\ostrat}({\playset})$ for the probability of a
measurable set of plays $\playset$ starting from~$\state_0$.
It is initially defined for the cylinder sets generated by the histories and then extended to the sigma-algebra by Carath\'eodory's unique extension
theorem~\cite{billingsley-1995-probability}.
Given a random variable $V : Z^\omega \to \mathbb{R}$, we will write $\expectval_{\game,\state_0,\zstrat,\ostrat}(V)$ for the expectation of~$V$ w.r.t.~$\probm_{\game,\state_0,\zstrat,\ostrat}$.
We may drop~$\game$ from the subscript when it is understood.

\subsection*{Objectives}
We consider reachability and safety objectives.
Given a set $\reachset \subseteq \states$ of states, the \emph{reachability} objective $\reach{\reachset}$ is the set of plays that visit $\reachset$ at least once, i.e., $s_h \in \reachset$ holds for some history~$h$ that is a prefix of the play.
The dual \emph{safety} objective~$\avoid{\reachset} \eqdef Z^\omega \setminus \reach{\reachset}$ consists of the plays that never visit~$T$.
We can and will assume that $T = \{\top\}$ holds for a sink state $\top \in S$ and write $\reach{\top}$ for $\reach{\{\top\}}$.
Similarly, we assume that there is another sink state $\bot \in S$, with $\bot \ne \top$, and write $\avoid{\bot}$ for~$\avoid{\{\bot\}}$.
Max attempts to maximize the probability of achieving the given objective (usually $\reach{\top}$ or $\avoid{\top}$), whereas Min attempts to minimize it.

\subsection*{Value and Optimality}
For a game $\game$, objective $\playset$ and initial state~$\state_0$, the \emph{lower value} and \emph{upper value} of~$s_0$ are respectively defined as
\[
\valueoflower{\game,\playset}{\state_0} \eqdef \sup_{\zstrat \in \zstratset} \inf_{\ostrat \in \ostratset}
\probm_{\game,\state_0,\zstrat,\ostrat}({\playset}) \quad \text{and} \quad
\valueofupper{\game,\playset}{\state_0} \eqdef \inf_{\ostrat \in \ostratset} \sup_{\zstrat \in \zstratset}
\probm_{\game,\state_0,\zstrat,\ostrat}({\playset})\,.
\]
The inequality $\valueoflower{\game,\playset}{\state_0} \le \valueofupper{\game,\playset}{\state_0}$ holds by definition.
If $\valueoflower{\game,\playset}{\state_0} = \valueofupper{\game,\playset}{\state_0}$, then this quantity is called the
\emph{value}, denoted by $\valueof{\game,\playset}{\state_0}$.
For reachability objectives, like all Borel objectives, the value exists if all action sets are finite~\cite{Maitra-Sudderth:1998}, and even if for all states~$s$ we have that $A(s)$ is finite or $B(s)$ is finite~\cite[Theorem 11]{Flesch-Predtetchinski-Sudderth:2020}.
We always assume the latter, 
so that $\valueof{\game,\playset}{s_0}$ exists.
For $\eps \ge 0$, a Max strategy~$\zstrat$ is called \emph{$\eps$-optimal} from~$s_0$ if for all Min strategies~$\pi$ we have
$\probm_{\game,\state_0,\zstrat,\ostrat}({\playset})\ge \valueof{\game,\playset}{\state_0} - \eps$.
A $0$-optimal strategy is also called \emph{optimal}.
For~$E = \reach{\top}$ or $E = \avoid{\bot}$ we have
\begin{equation} \label{eq:value}
\valueof{\game,E}{s_0} \ = \ \sup_{\alpha \in \dist(A(s_0))} \inf_{\beta \in \dist(B(s_0))} \sca{p(s_0,\alpha,\beta)}{\valueo{\game,E}{}}\,.
\end{equation}
We note in passing that the equality in~\eqref{eq:value} also holds in the states~$\top, \bot$, as these are sink states.

\section{Martingales for Safety}

For the remainder of the paper we fix a game~$\game$ over a countable state space~$S$.
Whenever we make finiteness assumptions on~$S$ and the action sets, we state them explicitly.
Several times we use the following lemma, a consequence of the optional-stopping theorem for submartingales.

\begin{lemma} \label{lemma:mart}
Let $\sigma$ and~$\pi$ be Max and Min strategies, respectively.
Suppose $v : S \to [0,1]$ is a function with $v(\bot) = 0$ such that $\sca{p(s_h,\sigma(h),\pi(h))}{v} \ge v(s_h)$ holds for all histories~$h$.
Then $\probm_{s_0,\sigma,\pi}(\avoid{\bot}) \ge v(s_0)$ holds for all $s_0 \in S$.
\end{lemma}
\begin{proof}
Let $s_0 \in S$.
Define a sequence of random variables $V(0), V(1), \ldots$ with $V(i) : Z^\omega \to [0,1]$ by
\[
 V(i)((s_0,a_0,b_0) (s_1,a_1,b_1) \cdots) \ \eqdef \ v(s_i)\,.
\]
Similarly, denote by $H(i) : Z^\omega \to H_i$ the function that maps each random  play to its unique prefix in~$H_i$.
Recall that we can create random variables via composition of functions. Consider the function~$s: H \to S$ that maps each finite history~$h$ to its final state~$s_h$; this function composed with $H(i)$ is the  random variable $s_{H(i)}: Z^\omega \to S$. Moreover, the function~$v$ composed with  $s_{H(i)}$ is the random variable~$v(s_{H(i)})$.
Note that  $V(i) = v(s_{H(i)})$. For $i \ge 0$, let $\mathcal{F}_i$ be the sigma-algebra generated by the cylinder sets corresponding to the histories $h \in H_i$.
Then $V(i)$ and~$H(i)$ are $\mathcal{F}_i$-measurable.
By the assumption on~$v$ we have
\begin{align*}
\expectval_{s_0,\sigma,\pi} (V(i+1) \mid \mathcal{F}_i) \
& = \ \expectval_{s_0,\sigma,\pi} (v(s_{H(i+1)}) \mid \mathcal{F}_i) \\
& = \ \sum_{s \in S} p(s_{H(i)}, \sigma(H(i)), \pi(H(i))) v(s) \\
& = \ \sca{p(s_{H(i)}, \sigma(H(i)), \pi(H(i)))}{v} \\
& \ge \ v(s_{H(i)}) \ = \ V(i)\;.
\end{align*}
It follows that
$V(0), V(1), \ldots$ is a submartingale with respect to the filtrations~$\mathcal{F}_0, \mathcal{F}_1, \ldots$.
Moreover $|V(i)| \le 1$ holds for all~$i$. This boundedness condition together with
the optional-stopping theorem imply that, almost surely, $V(0), V(1), \ldots$ converges to a random variable $V(\infty)$ with
\[
 \expectval_{s_0,\sigma,\pi} (V(\infty)) \ \ge \ \expectval_{s_0,\sigma,\pi} (V(0)) \ = \ v(s_0)\,.
\]
On the other hand, since $\bot$~is a sink,
\begin{align*}
 \expectval_{s_0,\sigma,\pi} (V(\infty)) \
 & = \ \probm_{s_0,\sigma,\pi} (\avoid{\bot}) \cdot \expectval_{s_0,\sigma,\pi}(V(\infty) \mid \avoid{\bot}) \\
 & \ \  \mbox{} + \probm_{s_0,\sigma,\pi} (\reach{\bot}) \cdot \expectval_{s_0,\sigma,\pi}(V(\infty) \mid \reach{\bot}) \\
 & \le\ \probm_{s_0,\sigma,\pi}(\avoid{\bot}) \cdot 1 + \probm_{s_0,\sigma,\pi}(\reach{\bot}) \cdot 0 \\
 & =\ \probm_{s_0,\sigma,\pi}(\avoid{\bot})\,. \tag*{\qed}
\end{align*}
\end{proof}

\section{Optimal Safety}
We will construct memoryless strategies for reachability from strategies for safety.
The following proposition will be useful for this purpose and is of independent interest.

\begin{proposition} \label{proposition:safety-memoryless}
Suppose Max's action sets are finite.
Then Max has a memoryless strategy~$\sigma$ that is optimal for $\avoid{\bot}$ from every state.
That is, for all $s \in S$ and all Min strategies~$\pi$
\[
 \probm_{s,\sigma,\pi}(\avoid{\bot}) \ \ge\ \valueof{\avoid{\bot}}{s}\,.
\]
\end{proposition}
\begin{proof}
Since Max has only finite action sets, the supremum in~\eqref{eq:value} is taken over a compact set of mixed actions.
Therefore, it is a maximum; i.e, for every $s \in S$ there is a mixed action, say $\sigma(s) \in \dist(A(s))$, such that for all mixed Min actions $\beta \in \dist(B(s))$
\begin{equation}
\label{eq:safe}
 \valueof{\avoid{\bot}}{s} \ \le \ \sca{p(s,\sigma(s),\beta)}{\valueo{\avoid{\bot}}}\,.
\end{equation}
Extend~$\sigma$ to a memoryless Max strategy in the natural way, and let $\pi$ be an arbitrary Min strategy.
Then the function $\valueo{\avoid{\bot}}$ satisfies the conditions of \cref{lemma:mart}.
Thus, $\probm_{s,\sigma,\pi}(\avoid{\bot}) \ge \valueof{\avoid{\bot}}{s}$ holds for all~$s$.
\qed
\end{proof}

\begin{figure}[tbp]
\begin{center}
\begin{minipage}[l]{0.4\textwidth}

\begin{tabular}{c|cc}
   transition& \multirow{2}{*}{$s$} & \multirow{2}{*}{ $\bot$}  \\
    function~$p$&& \\
 \hline
 $p(s,a_i,b)$ & $1-2^{-i} \; $  &  $\;  2^{-i}$\\
  \end{tabular}
  \end{minipage}
\begin{minipage}[r]{0.5\textwidth}
\begin{tikzpicture}[scale=0.8,xscale=1.6,yscale=1.3]

\node[Grand] (s) at (0,0) {$s$};
\node[Grand] (bot) at (2.5,0) {$\bot$};

\draw[->] (s) edge node[pos=.3, below=.15cm,fill=blue!40!white,rectangle,inner sep=2]{$a_i$} node[pos=.5, below=.07cm,fill=red!40!white,diamond, inner sep=1]{$b$} node[pos=.75, below=.05cm,inner sep=2]{$:~2^{-i}$}  (bot);
\draw[->] (s) edge[loop above] node[above left=.15cm and .3cm,fill=blue!40!white,rectangle,inner sep=2]{$a_i$} node[above=.1cm,fill=red!40!white,diamond, inner sep=1]{$b$} node[above right=.15cm and .2cm,inner sep=2]{$:~1-2^{-i}$}  (s);
\draw[->] (bot) edge[loop above] (bot);
\end{tikzpicture}

\end{minipage}	
\end{center}	
\caption{A finite-state MDP where Max has no $\eps$-optimal memoryless strategy for safety objective~$\avoid{\bot}$. The state $\bot$ is a sink state.
The set $\{a_i \mid i\in \mathbb{N}\}$ of Max's actions at $s$ is countable, while Min only has a single action~$b$.}
\label{fig:safety}
\end{figure}
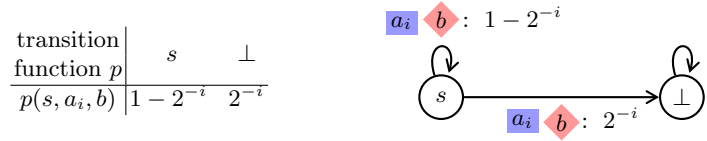

\begin{example}
The assumption that Max has finite action sets cannot be dropped from~\Cref{proposition:safety-memoryless}. This assumption is required  even for finite-state MDPs, i.e., when $S$~is finite and Min has only one action per state.
In fact, memoryless $\eps$-optimal strategies do not always exist for safety objectives. As an example, consider the  game depicted in
\Cref{fig:safety}, which was first discussed in~\cite{KMSW2017}.
The strategy that plays~$a_{i+k}$ at stage~$i$ is $\frac{1}{2^k}$-optimal for Max; indeed, the probability of reaching $\bot$ at stage~$i$ will be at most~$2^{-(i+k+1)}$. Hence, the probability of $\avoid{\bot}$ is at least $1-\sum_{i\in \mathbb{N}}2^{-(i+k+1)}= 1-\frac{1}{2^k}$, as required. This implies that $\valueof{\avoid{\bot}}{s}=1$.

Let $\sigma$ be some arbitrary memoryless strategy for Max. The probability that state~$\bot$ is not reached at stage~$i$ is  $(1-p(s,\sigma(s),b)(\bot))^i$. Clearly, the probability of the event $\avoid{\bot}$ is~$0$.
\qed
\end{example}

\section{$\eps$-Optimal Reachability} \label{section:eps-optimal-reach}

Max does not always have optimal strategies for reachability, even when $S$ and all action sets are finite.
\begin{example} \label{example:snowball}
Consider the  game depicted in
\Cref{fig:hiderun},
introduced in~\cite{KumarShiau}, also known as \emph{snowball} or \emph{hide-or-run} \cite{de2007concurrent}.
Following the intuition given in~\cite{de2007concurrent}, Max initially hides behind a bush (state~$s$) and his goal is to reach home (state~$\top$) without being hit by a snowball; Min is armed with a single snowball.

Below, in order to show that $\valueof{\reach{\top}}{s}=1$,
for all~$\eps$, such that $0< \eps<1$, we exhibit an $\eps$-optimal strategy for Max.
Given~$\eps$, consider the memoryless
 Max's strategy $\sigma$ defined by $\sigma(s)(hide)=1-\eps$ and
$\sigma(s)(run)=\eps$. After fixing $\sigma$ in the game, by \Cref{proposition:safety-memoryless}, Min has  optimal memoryless strategies~$\pi$ for $\avoid{\top}$ in the new game.
Let $\pi$ be a memoryless strategy for Min,  defined as
\[\pi(s)(throw)=x \qquad   \text{ and } \qquad \pi(s)(wait)=1-x.\]
for some $0\leq  x \leq 1$. Then
\[
\probm_{s,\sigma,\pi}(\reach{\top}) = \eps(1-x)+(1-\eps)x+(1-\eps)(1-x)\probm_{s,\sigma,\pi}(\reach{\top}).
\]
Solving this for $\probm_{s,\sigma,\pi}(\reach{\top})$ yields
\[
  \probm_{s,\sigma,\pi}(\reach{\top}) \ = \ 1- \frac{x\eps}{\eps+x(1-\eps)} \ \geq \ 1- \frac{x\eps}{x\eps+x(1-\eps)} \ = \ 1-\eps\,,\]
implying that $\sigma$ is an $\eps$-optimal strategy for Max.

 A straightforward argument shows that there is no optimal strategy for Max in
this game. If Max always plays hide, Min can wait forever. If not,
assume that at some stage Max plays run with probability $\eps>0$; then Min
would throw to reach $\bot$ with a positive probability.
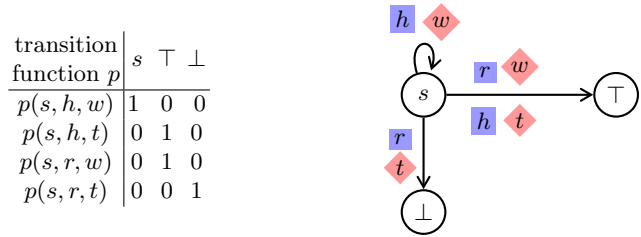
\begin{figure}[tbp]
\begin{center}
\begin{minipage}[l]{0.4\textwidth}

\begin{tabular}{c|ccc}
   transition&  \multirow{2}{*}{$s$} &\multirow{2}{*}{$\top$} & \multirow{2}{*}{$\bot$} \\
 function~$p$ &&& \\
 \hline
$p(s,h,w)$ & $1 \;$ & $\; 0 \; $  & $\; 0$\\
$p(s,h,t)$ & $0  $ & $1$& $0$\\
$p(s,r,w)$ & $0$ & $1$ &$0$\\
$p(s,r,t)$ & $0$ & $0$&$1$\\
  \end{tabular}
  \end{minipage}
\begin{minipage}[r]{0.5\textwidth}
\begin{tikzpicture}[scale=0.8,xscale=1.6,yscale=1.3]

\node[Grand] (s) at (0,0) {$s$};
\node[Grand] (top) at (2,0) {$\top$};
\node[Grand] (bot) at (0,-1.5) {$\bot$};
\draw[->] (s) edge node[pos=.3, left=.15cm,fill=blue!40!white,rectangle,inner sep=2]{$r$} node[pos=.7, left=.07cm,fill=red!40!white,diamond, inner sep=1]{$t$}   (bot);
\draw[->] (s) edge node[above left=.15cm and .3cm,fill=blue!40!white,rectangle,inner sep=2]{$r$} node[above=.1cm,fill=red!40!white,diamond, inner sep=1]{$w$}
node[below left=.15cm and .3cm,fill=blue!40!white,rectangle,inner sep=2]{$h$} node[below=.1cm,fill=red!40!white,diamond, inner sep=1]{$t$}(top);

\draw[->] (s) edge[loop above] node[above left=.15cm and .1cm,fill=blue!40!white,rectangle,inner sep=2]{$h$} node[above right=.16cm and .1cm,fill=red!40!white,diamond, inner sep=1]{$w$}   (s);
\end{tikzpicture}
\end{minipage}	
\end{center}	
\caption{The snowball game where Max has no optimal  strategy for~$\reach{\top}$. The states~$\top$ and $\bot$ are sink states.
Max's action set  at $s$ is $\{hide,run\}$, shown as $h$ and $r$ in the figure, while Min's action set is~$\{wait,throw\}$, shown as $w$ and $t$.}
\label{fig:hiderun}
\end{figure}
\qed
\end{example}

In this section we prove the following proposition.

\begin{proposition} \label{proposition:reach-eps}
Suppose that Min's action sets are finite.
Then for every $\eps > 0$ and every finite subset of states $S_0 \subseteq S$ Max has a memoryless strategy~$\sigma$ that is $\eps$-optimal for~$\reach{\top}$ from all $s_0 \in S_0$ and for each $s \in S$ the support of~$\sigma(s)$ is finite.
\end{proposition}
\begin{remark}
The assumption that Min's action sets are finite can be replaced by the assumption that $S$ is finite.
The proof requires an extension of \cref{lemma:mart}.
\end{remark}

\Cref{proposition:reach-eps} does not carry over to countably infinite
reachability games with infinite action sets for Min; see \Cref{sec:conclusion}.


For the rest of the section we assume that Min's action sets are finite.
Towards a proof of \cref{proposition:reach-eps} we consider, for $n \ge 0$, the horizon-restricted reachability objective $\reachn{n}{\top} \subseteq \reach{\top}$ consisting of the plays that reach~$\top$ within at most $n$ steps.
For $s \in S$ write $\valueof{n}{s}$ for $\valueof{\reachn{n}{\top}}{s}$.
We have $\valueof{0}{\top} = 1$ and $\valueof{0}{s} = 0$ for all $s \ne \top$.
For all $n \ge 0$, since $\reachn{n}{\top} \subseteq \reachn{n+1}{\top}$, we have for all $s \in S$
\begin{align*}
 \valueof{n}{s} \ \le \ \valueof{n+1}{s} \ &= \ \sup_{\alpha \in \dist(A(s))} \inf_{\beta \in \dist(B(s))} \sca{p(s,\alpha,\beta)}{\valueo{n}} \\
                    &= \ \sup_{\alpha \in \dist(A(s))} \min_{b \in B(s)} \sca{p(s,\alpha,b)}{\valueo{n}} \,,
\end{align*}
as Min's action sets are finite.
For all $s \in S$, since $\valueof{0}{s} \le \valueof{1}{s} \le \ldots \le 1$, there is a limit $\valueof{\infty}{s} \eqdef \lim_{n \to \infty} \valueof{n}{s}$.

\begin{lemma} \label{lemma:v_infty-fixed}
Suppose that Min's action sets are finite.
For all $s \in S$ we have
\[
 \valueof{\infty}{s} \ = \ \sup_{\alpha \in \dist(A(s))} \min_{b \in B(s)} \sca{p(s,\alpha,b)}{\valueo{\infty}} \,.
\]
\end{lemma}
\begin{proof}
Let $s \in S$.
Towards the ``$\le$'' inequality, for any~$n$ we have
\begin{align*}
 \valueof{n+1}{s} \ &= \  \sup_{\alpha \in \dist(A(s))} \min_{b \in B(s)} \sca{p(s,\alpha,b)}{\valueo{n}} \\
                    &\le\ \sup_{\alpha \in \dist(A(s))} \min_{b \in B(s)} \sca{p(s,\alpha,b)}{\valueo{\infty}}\,.
\end{align*}
Thus, $\valueof{\infty}{s}  \le \sup_{\alpha \in \dist(A(s))} \min_{b \in B(s)} \sca{p(s,\alpha,b)}{\valueo{\infty}}$.

Towards the ``$\ge$'' inequality, for any fixed~$\alpha$ and~$b$ the function $\sca{p(s,\alpha,b)}{\cdot} : \mathbb{R}^S \to \mathbb{R}$ is a linear map and, thus, continuous.
Hence, for any $\alpha$
\begin{align*}
   \lim_{n \to \infty} \min_{b \in B(s)} \sca{p(s,\alpha,b)}{\valueo{n}} \
   &= \ \min_{b \in B(s)} \lim_{n \to \infty} \sca{p(s,\alpha,b)}{\valueo{n}} \\
   &= \ \min_{b \in B(s)} \sca{p(s,\alpha,b)}{\valueo{\infty}}\,.
\end{align*}
It follows that
\begin{align*}
 \valueof{\infty}{s} \
 =\   \lim_{n \to \infty} \valueof{n+1}{s} \
 &=\   \lim_{n \to \infty} \sup_{\alpha \in \dist(A(s))} \min_{b \in B(s)} \sca{p(s,\alpha,b)}{\valueo{n}} \\
 &\ge\ \sup_{\alpha \in \dist(A(s))} \lim_{n \to \infty} \min_{b \in B(s)} \sca{p(s,\alpha,b)}{\valueo{n}} \\
 &=\   \sup_{\alpha \in \dist(A(s))} \min_{b \in B(s)} \sca{p(s,\alpha,b)}{\valueo{\infty}}\,. \tag*{\qed}
\end{align*}
\end{proof}

Now the following lemma follows from \cref{lemma:mart}.
\begin{lemma} \label{lemma:v_infty=val}
Suppose that Min's action sets are finite.
Then $\valueo{\reach{\top}} = \valueo{\infty}$.
\end{lemma}
\begin{proof}
Since $\reach{\top} \supseteq \reachn{n}{\top}$, we have $\valueo{\reach{\top}} \ge \valueo{n}$.
Towards the other inequality, define a function $v : S \to [0,1]$ by $v(s) \eqdef 1 - \valueof{\infty}{s}$.
We have $v(\top) = 1 - \valueof{\infty}{\top} = 0$.
Let $\sigma$ be an arbitrary Max strategy.
For every history~$h$, let $b_h \in \argmin_{b \in B(s_h)} \sca{p(s_h,\sigma(h),b)}{\valueo{\infty}}$.
Define a Min strategy~$\pi$ with $\pi(h)(b_h) = 1$ for all histories~$h$.
For all histories~$h$ we have
\begin{align*}
 \sca{p(s_h,\sigma(h),\pi(h))}{v} \
 &= \   1 - \sca{p(s_h,\sigma(h),\pi(h))}{\valueo{\infty}} && \text{from the definition of~$v$} \\
 &\ge\  1 - \valueof{\infty}{s_h} && \text{\cref{lemma:v_infty-fixed}} \\
 & =\   v(s_h) && \text{definition of~$v$.}
\end{align*}
By \cref{lemma:mart}, for all $s \in S$
\[
 \probm_{s,\sigma,\pi}(\reach{\top}) \ =\ 1 - \probm_{s,\sigma,\pi}(\avoid{\top}) \ \le\ 1 - v(s) = \valueof{\infty}{s}\,.
\]
Since $\sigma$ was arbitrary, we conclude that $\valueo{\reach{\top}} \le \valueo{\infty}$.
\qed
\end{proof}

Let us introduce an operation on games which makes transitions ``leak'' to~$\bot$.
The intention is to ``reduce'' reachability to safety, in the sense that in a leaky game, $\avoid{\bot}$ is included in (and, hence, equal to) $\reach{\top}$ up to a measure-zero set of plays.
We set up the leaky game so that Max has an optimal strategy for~$\avoid{\bot}$, which, according to \cref{proposition:safety-memoryless}, can be chosen to be memoryless.

For technical reasons, which will manifest themselves later, we associate the leaks with Min actions.
Concretely, for a state~$s$, a Min action $b \in B(s)$, and some $\eps > 0$, by \emph{making~$b$ leak~$\eps$} we refer to obtaining from~$p$ another transition function $\check{p}$ by setting, for all $a \in A(s)$ and $t \in S$,
\[
 \check{p}(s,a,b)(t) \ \eqdef \ \begin{cases} (1-\eps) p(s,a,b)(t) & \text{if } t \ne \bot \\ (1-\eps) p(s,a,b)(\bot) + \eps & \text{if } t= \bot\,. \end{cases}
\]
Intuitively, a fraction of~$\eps$ of the probability mass leaks to~$\bot$ whenever $b$~is taken.
We use leaks to prove \cref{proposition:reach-eps}.

\begin{proof}[of \cref{proposition:reach-eps}]
Let $\eps > 0$, and let $S_0 \subseteq S$ be finite.
Choose $\eps_1, \eps_2, \eps_3 > 0$ such that $\eps_1 + \eps_2 + \eps_3 = \eps$.
For each $s \in S$ choose $a(s) \in A(s)$.
By \cref{lemma:v_infty=val}, since $S_0$ is finite, there is $n \ge 0$ such that $\valueof{n}{s} \ge \valueof{\reach{\top}}{s} - \eps_1$ holds for all $s \in S_0$.
Inductively define (in general non-memoryless) Max strategies $\sigma_0, \ldots, \sigma_n$ as follows.
For each history~$h$, define $\sigma_0(h)(a(s_h)) \eqdef 1$.
For $i \in \{0, \ldots, n-1\}$ and each state $s \in S$ we have $\sup_{\alpha \in \dist(A(s))} \min_{b \in B(s)} \sca{p(s,\alpha,b)}{\valueo{i}} = \valueof{i+1}{s}$; thus we can define $\sigma_{i+1}(s)$ so that $\min_{b \in B(s)} \sca{p(s,\sigma_{i+1}(s),b)}{\valueo{i}} \ge \valueof{i+1}{s} - \frac{\eps_2}{n}$.
Since $\frac{\eps_2}{n} > 0$, we can assume without loss of generality that the support of $\sigma_{i+1}(s)$ is finite.
Finally, define $\sigma_{i+1}(z h) \eqdef \sigma_i(h)$ for all $z \in Z$ and $h \in H$.

We show inductively that for all $i \in \{1, \ldots, n\}$ and all $s \in S$ 
\begin{equation} \label{eq:proposition:reach-eps-1}
 \inf_{\pi \in \ostratset} \probm_{\game,s,\sigma_i,\pi}(\reachn{i}{\top}) \ \ge \ \valueof{i}{s} - i \cdot \frac{\eps_2}{n} \,.
\end{equation}
This is immediate for $i=0$.
For $i \in \{0, \ldots, n-1\}$ we have
\begin{align*}
& \inf_{\pi \in \ostratset} \probm_{\game,s,\sigma_{i+1},\pi}(\reachn{i+1}{\top}) \\
& \ge\ \min_{b \in B(s)} \sca{p(s,\sigma_{i+1}(s),b)}{\inf_{\pi \in \ostratset} \probm_{\game,\cdot,\sigma_i,\pi}(\reachn{i}{\top})} \\
& \ge\ \min_{b \in B(s)} \sca{p(s,\sigma_{i+1}(s),b)}{\valueo{i}} - i \cdot \frac{\eps_2}{n} && \text{induction hypothesis} \\
& \ge\ \valueof{i+1}{s} - \frac{\eps_2}{n} - i \cdot \frac{\eps_2}{n} && \text{definition of } \sigma_{i+1}(s)\,,
\end{align*}
proving~\eqref{eq:proposition:reach-eps-1}.

Now in every state except~$\top$, make each Min action leak~$\frac{\eps_3}{n}$.
Further, remove any Max actions that $\sigma_n$ never takes.
Call the resulting game~$\game^-$.
Then Max's action sets in~$\game^-$ are finite and we have for all $s_0 \in S$ and all Min strategies~$\pi$
\begin{align*}
 & \probm_{\game^-,s,\sigma_n,\pi}(\avoid{\bot}) \\
 & \ge \ \probm_{\game^-,s,\sigma_n,\pi}(\reachn{n}{\top}) && \avoid{\bot} \supseteq \reachn{n}{\top} \\
 & \ge \ \probm_{\game,s,\sigma_n,\pi}(\reachn{n}{\top}) - \eps_3 && \text{at most $n \cdot \frac{\eps_3}{n}$ leaks in $n$ steps} \\
 & \ge \ \valueof{n}{s} - \eps_2 - \eps_3 && \text{by~\eqref{eq:proposition:reach-eps-1}} \\
 & \ge \ \valueof{\game,\reach{\top}}{s} - \eps_1 - \eps_2 - \eps_3 && \text{choice of $n$} \\
 & =   \ \valueof{\game,\reach{\top}}{s} - \eps && \text{choice of $\eps_1, \eps_2, \eps_3$.}
\end{align*}
Thus, for all $s \in S_0$
\begin{equation} \label{eq:proposition:reach-eps-2}
 \valueof{\game^-,\avoid{\bot}}{s} \ \ge \ \valueof{\game,\reach{\top}}{s} - \eps\,.
\end{equation}
Due to the leaks, in~$\game^-$ the events $\reach{\top}$ and $\avoid{\bot}$ coincide up to measure zero for all Max and all Min strategies.
Since Max's action sets in~$\game^-$ are finite, by \cref{proposition:safety-memoryless} Max has a memoryless strategy~$\sigma$ that is optimal for $\avoid{\bot}$ from every state.
Thus, for all $s \in S_0$ and all Min strategies~$\pi$
\begin{align*}
  \probm_{\game,s,\sigma,\pi}(\reach{\top}) \
 & \ge \ \probm_{\game^-,s,\sigma,\pi}(\reach{\top}) \\
 &  =  \ \probm_{\game^-,s,\sigma,\pi}(\avoid{\bot}) && \text{as argued above} \\
 & \ge \ \valueof{\game^-,\avoid{\bot}}{s} && \text{$\sigma$ is optimal} \\
 & \ge \ \valueof{\game,\reach{\top}}{s} - \eps && \text{by \eqref{eq:proposition:reach-eps-2}.}
\end{align*}
That is, the memoryless strategy $\sigma$ is $\eps$-optimal from all $s \in S_0$, and for each $s \in S$ the support of $\sigma(s)$ is finite.
\qed
\end{proof}

\section{Optimal Reachability}

In this section we show the following result, a generalization of Bordais et al.~\cite[Theorem~28]{BordaisB022}, via a different proof.

\begin{theorem} \label{theorem:main}
Suppose that $S$ and Min's action sets are finite.
Then for every $\eps > 0$ Max has a memoryless strategy~$\sigma$ that is $\eps$-optimal for~$\reach{\top}$ from all states and optimal from all states from which Max has any optimal strategy.
\end{theorem}

\begin{figure}[tbp]
\begin{center}
\begin{minipage}[l]{0.35\textwidth}

\begin{tabular}{c|cc}
   transition& \multirow{2}{*}{$s_1$} & \multirow{2}{*}{ $\bot$}  \\
    function~$p$&& \\
 \hline
 $p(s_0,a,b_q)$ & $q \; $  &  $\;  1-q$\\
  \end{tabular}

  \end{minipage}
\begin{minipage}[r]{0.6\textwidth}
\begin{tikzpicture}[scale=0.8,xscale=1.6,yscale=1.3]

\draw [draw=black, dashed] (-.7,1.4) rectangle (2.5,-2);

\node[Grand] (s0) at (-3,0) {$s_0$};
\node[Grand] (s) at (0,0) {$s_1$};
\node[Grand] (top) at (2,0) {$\top$};
\node[Grand] (bot) at (0,-1.5) {$\bot$};

\draw[->] (s) edge node[pos=.3, left=.15cm,fill=blue!40!white,rectangle,inner sep=2]{$r$} node[pos=.7, left=.07cm,fill=red!40!white,diamond, inner sep=1]{$t$}   (bot);
\draw[->] (s) edge node[above left=.15cm and .3cm,fill=blue!40!white,rectangle,inner sep=2]{$r$} node[above=.1cm,fill=red!40!white,diamond, inner sep=1]{$w$}
node[below left=.15cm and .3cm,fill=blue!40!white,rectangle,inner sep=2]{$h$} node[below=.1cm,fill=red!40!white,diamond, inner sep=1]{$t$}(top);

\draw[->] (s) edge[loop above] node[above left=.15cm and .1cm,fill=blue!40!white,rectangle,inner sep=2]{$h$} node[above right=.16cm and .1cm,fill=red!40!white,diamond, inner sep=1]{$w$}   (s);

\draw[->] (s0) edge node[pos=.3, above=.2cm,fill=blue!40!white,rectangle,inner sep=2]{$a$} node[pos=.45, above=.07cm,fill=red!40!white,diamond, inner sep=0]{$b_q$} node[pos=.6, above=.16cm,inner sep=2]{$:q$}  (s);

\draw[->] (s0) edge node[pos=.3, below=.2cm,fill=blue!40!white,rectangle,inner sep=2,rotate=-21]{$a$} node[pos=.45, below=.07cm,fill=red!40!white,diamond, inner sep=0,rotate=-18]{$b_q$} node[pos=.65, below=.15cm,inner sep=2,rotate=-18]{$:1-q$}  (bot);

\end{tikzpicture}
\end{minipage}	
\end{center}	
\caption{A game where Max has no \emph{memoryless} optimal  strategy for~$\reach{\top}$. The states~$\top$ and $\bot$ are sink states.
Min's action set  at $s_0$ is $\{b_q \mid q \in (\frac12,1) \cap \mathbb{Q}\}$, while Max has a single action~$a$ at $s_0$.
The game in dashed box is the snowball game from~\Cref{fig:hiderun}.  }
\label{fig:newsnowball}
\end{figure}
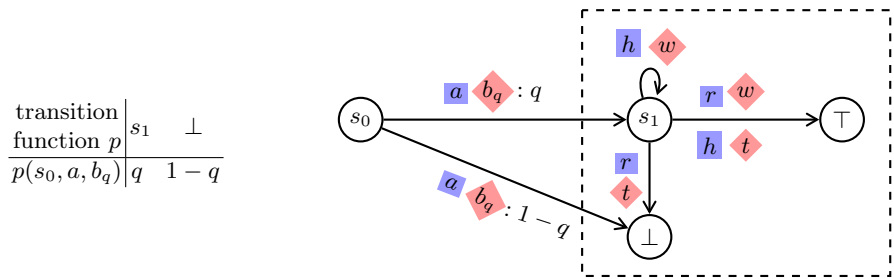

The assumption that Min's action sets are finite cannot be dropped, as the following example shows.

\begin{example}
Consider the game depicted in~\Cref{fig:newsnowball}, wherein
state~$s_1$ is the start state of the snowball game from \cref{example:snowball}. Recall that
 $\valueof{\reach{\top}}{s_1} = 1$, but Max does not have an optimal strategy from~$s_1$.
By this, and the fact that Min can choose $b_q$ for $q$ arbitrary close to $\frac12$, we deduce that  $\valueof{\reach{\top}}{s_0} = \frac12$. Observe that when  taking action $b_q$ at $s_0$, Min is increasing the  value from $\frac12$ to $q > \frac12$. Indeed, we claim that
Max has an optimal strategy from~$s_0$. If Max observes which action $b_q$ Min takes in~$s_0$, and  plays a $(q-\frac12)$-optimal strategy from~$s_1$, the probability of reaching~$\top$ will be at least $q \cdot (1 - (q-\frac12)) \ge q - (q-\frac12) = \frac12$.
But Max has to ``remember'' Min's value increase in order to play sufficiently well in~$s_1$.
So Max does not have a memoryless optimal strategy from~$s_0$.
In fact, a step counter (discrete clock) would not help Max either, since the
step counter value does not contain any information about $q$.
\qed
\end{example}

\begin{figure}[tbp]
\begin{center}
\includegraphics[width=\textwidth, angle=0]{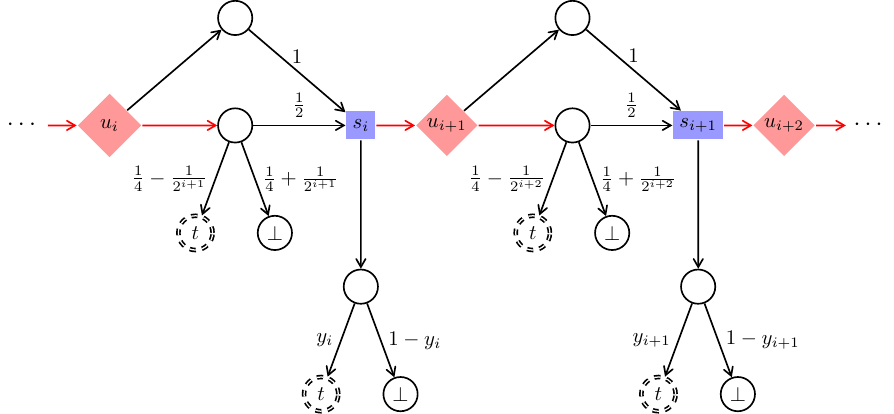}
\end{center}
\caption{A finitely-branching turn-based reachability game $\game$ with
  initial state $u_1$, where optimal Max strategies cannot be memoryless (and cannot  even be Markov).
  For clarity, we have drawn several copies of the target state $t$.
The number $y_i$ is defined as $\frac12 - \frac1{2^{i+1}}$.}
\label{fig:turn-fb-optmax-lower}
\end{figure}

The assumption that $S$ is finite cannot be dropped either,
not even for the subclass of turn-based games,
as the following example shows.

\begin{example}
Consider the finitely-branching turn-based reachability game depicted in
\Cref{fig:turn-fb-optmax-lower} (from \cite[Section 9]{KMSTW:DGA}).
The initial state is $u_1$, states $u_i$ are Min-controlled, states
$s_i$ are Max-controlled and $t$ is the target state.

At every state $u_i$, Min can either go right (red transition) or go up.
At every state $s_i$, Max can either go right (red transition) or go down.
It is easy to show that $\valueof{\game}{s_i} > \valueof{\game}{u_i}$ for all $i$.
Thus the only locally optimal Min move is to go right.
Similarly, $\valueof{\game}{s_i} > y_i$, and thus the only locally optimal Max move is to go right.

Max has an optimal strategy from $u_1$ (and also from every other state)
as follows. First, always go right. If Min ever goes up at some $u_i$
then go down at~$s_i$. Intuitively, by going up, Min increases the value
(i.e., gives a ``gift'' to Max). By going down at $s_i$, Max plays
sub-optimally locally, but realizes most of the value of $s_i$, i.e., loses
less than Min's previous gift.

However, no memoryless strategy (and not even any Markov strategy, i.e., any strategy that uses only a step counter) can be
optimal for Max from $u_1$.
First, a step counter gives no advantage to Max, since in this example
the step counter is implicit in the current state anyway.
For any memoryless Max strategy, there are to cases.
Either this strategy never goes down with any positive probability at any state $s_i$.
Then Min can avoid the target $t$ completely by always playing up,
and thus this Max strategy is not optimal.
Otherwise, let $s_j$ be first state where Max goes down with some
positive probability.
This Max strategy is not optimal either, since Min can always go
right (locally optimal), and Max's choice at $s_j$ is locally sub-optimal.

See \cite[Section 9]{KMSTW:DGA} for a formal proof of this example.
It is also shown there that,
in every countably infinite finitely-branching turn-based reachability game,
optimal Max strategies (if they exist at all) can be chosen as
deterministic and using a step counter plus 1 bit of public memory.
So the example above is tight.
\qed
\end{example}

Now we prove \cref{theorem:main}.
Let $S$ and Min's action sets be finite.
We partition the finite state space~$S$ into $S_0, S_1$ so that $S_0$~contains exactly the states from which Max has an optimal strategy for~$\reach{\top}$, and $S_1$~contains exactly the states from which Max does not have an optimal strategy.
Clearly, $\top, \bot \in S_0$.

By the finiteness of Min's action sets and \eqref{eq:value}, for every value-$0$ state~$s$ and every mixed Max action $\alpha \in \dist(A(s))$ Min has an action~$b$ that keeps the game (surely) in a value-$0$ state.
Therefore, an optimal Max strategy cannot rely on Min ``gifts'' in value-$0$ states.
Hence, we can assume without loss of generality that $\bot$~is the only value-$0$ state.
We write $S_0^+ \eqdef S_0 \setminus \{\bot\}$, so that all states in $S_0^+$ have a positive value.

The challenge is that the play can ``cross over'' between $S_0$ and~$S_1$.
Our approach is as follows.
First we fix a memoryless strategy~$\sigma_0$ on~$S_0$ that is optimal for Max as long as Min responds optimally.
Then, building on the idea of the previous section, we add ``leaks'' to the game so that
\begin{itemize}
\item in~$S_0$ the leaks do not decrease the value and Max still has an optimal strategy;
\item in~$S_1$ the leaks decrease the value only by little.
\end{itemize}
After having fixed~$\sigma_0$ on~$S_0$ and introduced suitable leaks, an optimal safety strategy for $\avoid{\bot}$ will serve as the memoryless strategy~$\sigma$ claimed in \cref{theorem:main}.
This strategy extends~$\sigma_0$ to the whole state space and is optimal from~$S_0$ (no matter what Min does).

Below, to avoid clutter whenever we write $\valueof{}{s}$ we mean $\valueof{\game, \reach{\top}}{s}$ where $\game$ is the original game, even when other auxiliary games are discussed.

We start by defining a memoryless strategy, $\sigma_0$, only on~$S_0$ so that $\sigma_0$~is optimal from~$S_0$ as long as Min ``does not increase the value''.

For a state $s \in S_0$, call a mixed Max action $\alpha \in \dist(A(s))$ \emph{optimality-preserving} if for all Min actions $b \in B(s)$ we have $\sca{p(s,\alpha,b)}{\valueo{}} \ge \valueof{}{s}$ and if $\sca{p(s,\alpha,b)}{\valueo{}} = \valueof{}{s}$ then the support of $p(s,\alpha,b)$ is a subset of~$S_0$.
Note that every optimal Max strategy from~$S_0$ is optimality-preserving at least in the first step.
Therefore, every state in~$S_0$ has an optimality-preserving mixed Max action.

For a state $s \in S_0$ and an optimality-preserving mixed action~$\alpha$, call a Min action $b \in B(s)$ \emph{value-preserving} if $\sca{p(s,\alpha,b)}{\valueo{}} = \valueof{}{s}$ and \emph{value-increasing} otherwise (i.e., $\sca{p(s,\alpha,b)}{\valueo{}} > \valueof{}{s})$.

We define inductively a non-decreasing sequence $R(0), R(1), R(2), \ldots$ of subsets of~$S_0$.
Define $R(0) \eqdef \{\top\}$.
For every $n \ge 0$, define $R(n+1)$ as $R(n)$ union the set of those states
$s \in S_0$ from which Max has an optimality-preserving mixed action,
say~$\alpha_s$, and a number $\delta_s > 0$ such that for all value-preserving
Min actions $b \in B(s)$ we have $p(s,\alpha_s,b)(R(n)) \ge
\delta_s$. We note in passing that the existence of~$\delta_s>0$
is guaranteed due to finiteness of Min's action set. Informally speaking,
$R(n+1)$ consists of those states that are already in~$R(n)$ or have an optimality-preserving mixed action~$\alpha_s$ so that in the next step $R(n)$ is entered with a positive probability unless Min takes a value-increasing action.
Defining $R \eqdef R(|S|-1)$, since $S$ is finite we have $R = R(|S|-1) = R(|S|) = R(|S|+1) =
\ldots$.
Moreover, we can  define $\delta \eqdef \min_{s \in R} \delta_s > 0$.

\begin{lemma} \label{lemma:R=S0+}
We have $R = S_0^+$.
\end{lemma}
\begin{proof}
By an easy induction, $\bot \not\in R$.
Hence $R \subseteq S_0^+$.
Towards the reverse inclusion, suppose there is $s \in S_0 \setminus R$.
It suffices to show that $s = \bot$.
Let $\sigma$ be an arbitrary optimal Max strategy starting from $s$.
Note that $\sigma(s)$ is optimality-preserving.
Since $s \not\in R = R(|S|)$, there is a value-preserving Min action $b \in B(s)$ such that the support of $p(s,\sigma(s),b)$ does not overlap with $R(|S|-1) = R$.
That is, the support of $p(s,\sigma(s),b)$ is a subset of $S_0 \setminus R$ and, thus, does not contain~$\top$.
Since $\sigma$ is optimal, $\sigma$~has to play optimally---and, thus, has to use optimality-preserving mixed actions---as long as Min has taken only value-preserving actions.
Then it follows inductively that there is a deterministic Min strategy (playing~$b$ in the very first step) that keeps the play in $S_0 \setminus R$ forever.
In particular, $\top$ is not reached.
Since the optimal strategy~$\sigma$ was arbitrary, we have $\valueof{}{s} = 0$; i.e., $s = \bot$.
\qed
\end{proof}

Define $\sigma_0$ to be the memoryless Max strategy on~$S_0$ with $\sigma_0(s) = \alpha_s$ for all $s \in S_0$, where $\alpha_s$ is from the definition of~$R(n)$.
For $s \in S_0$, call $b \in B(s)$ \emph{value-preserving} (or \emph{value-increasing}) if $b$~is value-preserving (or value-increasing, respectively) for $s$ and~$\sigma_0(s)$.

\begin{lemma} \label{lemma:stuck-in-S0+}
For any Max strategy that plays $\sigma_0$ on~$S_0$ and any Min strategy, almost all plays that eventually remain in~$S_0^+$ and eventually contain only value-preserving Min actions reach~$\top$.
\end{lemma}
\begin{proof}
Let $\sigma$ be a Max strategy that plays~$\sigma_0$ on~$S_0$, and let $\pi$ be a Min strategy.
Since the set of (finite) histories is countable, it suffices to consider the set of plays which always remain in~$S_0^+$ and contain only value-preserving Min actions.
Let $s \in S_0^+$.
By \cref{lemma:R=S0+}, we have $s \in R$.
Let $n \le |S|-1$ be the smallest~$n$ with $s \in R(n)$.
We have $p(s,\sigma_0(s),b)(R(n-1)) \ge \delta$ for all value-preserving Min actions~$b$.
It follows that at any time the probability that in at most $|S|-1$ steps $R(0) = \{\top\}$ is reached is at least $\delta^{|S|-1} > 0$.
Thus, almost surely $\top$~is eventually reached.
\qed
\end{proof}

Since $S_0$ and Min's action sets are finite,
we can choose an $\eps>0$ small enough so that for all $s \in S_0$ and all
value-increasing Min actions~$b$ we have
\begin{equation}
\label{eq:cond-2ep}
	\sca{p(s,\sigma_0(s),b)}{\valueo{}}
\ge \valueof{}{s} + 2 \eps
\end{equation}
So any value increase by a Min action is at least~$2 \eps$.

Define a game~$\game_1$ with state space $S_1 \cup \{\top, \bot\}$.
Its transition function~$p_1$ is obtained from~$p$ by redirecting probability mass away from $S_0 \setminus \{\top, \bot\}$ as follows.
Each transition to a state $s \in S_0 \setminus \{\top, \bot\}$ is redirected to~$\top$ with probability $\valueof{}{s}$ and to~$\bot$ with probability $1-\valueof{}{s}$.
Then the value of each state in~$\game_1$ is equal to its value in~$\game$.

As in the proof of \cref{proposition:reach-eps} we obtain from~$\game_1$ a ``leaky'' version, $\game_1^-$, with transition function~$p_1^-$, such that Max's action sets are finite
 and, defining $v(s) \eqdef \valueof{\game_1^-,\avoid{\bot}}{s}$ for all $s \in S_1 \cup \{\top, \bot\}$, we have
\begin{equation}\label{eq:star}
v(s) \ \ge \ \valueof{}{s} - \eps \qquad \text{for all } s \in S_1.
\end{equation}
By \cref{proposition:safety-memoryless} Max has a memoryless optimal (for safety) strategy~$\sigma_1$ in~$\game_1^-$.
Thus, for all $s \in S_1$ and all $b \in B(s)$, as in~\Cref{eq:safe}, we have

\[
\sca{p_1^-(s,\sigma_1(s),b)}{v} \ \ge \ v(s)\,.
\]

Next we define a ``leaky'' version, $\game^-$, of the whole game~$\game$ with the whole state space~$S$ and a transition function~$p^-$.
On~$S_1$, function~$p^-$ is defined similarly to~$p_1^-$, except that the probability mass that was directed away from $S_0 \setminus \{\top, \bot\}$ in the definition of $\game_1$ is now not directed away and instead enters~$S_0$ as originally.
However, the leaks in every transition from~$S_1$ are also present in~$p^-$.
Extend the function~$v$ to the whole state space~$S$ by defining $v(s) \eqdef \valueof{}{s}$ for all $s \in S_0$.
Then, for all $s \in S_1$ and all $b \in B(s)$
\[
\sca{p^-(s,\sigma_1(s),b)}{v} \ = \ \sca{p_1^-(s,\sigma_1(s),b)}{v} \ \ge \ v(s)\,.
\]

On~$S_0$, the transition function~$p^-$ is obtained from~$p$ by making value-increasing Min actions leak~$\eps$, in the sense defined in \cref{section:eps-optimal-reach}.
For all $s \in S_0$ and all value-increasing Min actions $b \in B(s)$

\begin{align*}
&\sca{p^-(s,\sigma_0(s),b)}{v} \\
&\ge \ \sca{p^-(s,\sigma_0(s),b)}{\valueo{}} - \eps && \text{by \eqref{eq:star} and the definition of~$v$ on~$S_0$} \\
&\ge \ \sca{p(s,\sigma_0(s),b)}{\valueo{}} - \eps - \eps && \text{from the definition of~$p^-$} \\
&\ge \ \valueof{}{s} + 2\eps - \eps - \eps && \text{by } \eqref{eq:cond-2ep}\\
&=   \ v(s) && \text{definition of $v$ on $S_0$.}
\end{align*}

For all $s \in S_0$ and all value-preserving Min actions $b \in B(s)$
\begin{align*}
&\sca{p^-(s,\sigma_0(s),b)}{v} \\
&=\ \sca{p(s,\sigma_0(s),b)}{v} && \text{from the definition of $p^-$} \\
&=\ \sca{p(s,\sigma_0(s),b)}{\valueo{}} && \text{support of $p(s,\sigma_0(s),b)$ is a subset of~$S_0$} \\
&=\ \valueof{}{s} && \text{$b$ is value-preserving} \\
&=\ v(s) && \text{definition of $v$ on $S_0$.}
\end{align*}

Define the memoryless strategy~$\sigma$ by naturally combining $\sigma_0$ and~$\sigma_1$.
We have shown above that for all $s \in S$ and all $b \in B(s)$ we have $\sca{p^-(s,\sigma(s),b)}{v} \ge v(s)$.
From applying \cref{lemma:mart} we conclude that for all $s \in S$ and all Min strategies~$\pi$ we have
\[
\probm_{\game^-,s,\sigma,\pi}(\avoid{\bot}) \ \ge \ v(s)\,.
\]
In~$\game^-$, due to the leaks, almost all $\bot$-avoiding plays eventually remain in~$S_0^+$ and eventually have only value-preserving Min actions.
But by \cref{lemma:stuck-in-S0+}, almost all of these plays reach~$\top$.
Thus, we have for all $s \in S$ and all Min strategies~$\pi$
\[
    \probm_{\game^-,s,\sigma,\pi}(\reach{\top}) \
 = \ \probm_{\game^-,s,\sigma,\pi}(\avoid{\bot}) \
 \ge\ v(s)\,,
\]
which, by definition, equals $\valueof{}{s}$ if $s \in S_0$ and, by~\eqref{eq:star}, is at least $\valueof{}{s} - \eps$ if $s \in S_1$.
Since $\probm_{\game,s,\sigma,\pi}(\reach{\top}) \ge \probm_{\game^-,s,\sigma,\pi}(\reach{\top})$, this completes the proof of \cref{theorem:main}.

\section{Conclusion and Related Work}\label{sec:conclusion}

We have shown that, in finite reachability games with finite action sets for Min,
optimal Max strategies, where they exist, can
be chosen as memoryless randomized. However, this does not carry over to
countably infinite reachability games.

Intuitively, the reason for this is that Min can play sub-optimally
and give ``gifts'' to Max that increase the expected
value of the current state, but delay progress towards the target.
In countably infinite reachability games,
Min might give infinitely many smaller and smaller gifts and delay
progress indefinitely, unless Max uses memory to keep track of these gifts
in order to react correctly.

Finite reachability games are simpler, because gifts from Min to Max
are universally lower bounded in size, due to the finite state space and the
finiteness of Min's action sets. Therefore, Min cannot give infinitely many
gifts to Max, except in a nullset of the plays. Without such distracting
gifts, Max can make steady progress towards the target.
Moreover, if Min does give a gift once, then Max does not need to remember how
large it was, since it is universally lower bounded.

The existence of optimal strategies is also affected by the size of the state
space. While finite turn-based reachability games always admit optimal Max strategies
\cite{CONDON1992203}, even in countably infinite MDPs, optimal Max strategies
for reachability need not exist.
However, it was shown by Ornstein \cite[Thm.~B]{Ornstein:AMS1969} that $\eps$-optimal
Max strategies for reachability in countably infinite MDPs can be chosen
as memoryless and deterministic.
These strategies can even be made uniform, i.e., independent of the start
state. Moreover, if an optimal strategy does exist for Max in a
countable MDP, then there also exists one that is memoryless and deterministic
\cite[Prop.~B]{Ornstein:AMS1969}.

These results on countable MDPs do not carry over to countable 2-player
stochastic reachability games.
While Max always has $\eps$-optimal randomized memoryless  strategies
in countable concurrent reachability games with finite action sets
\cite[Corollary~3.9]{Secchi97}, these
strategies depend on the start state and cannot be made uniform
\cite{Raghavan-Nowak:1991}.
This non-uniformity even holds for the subclass of countable finitely branching turn-based
reachability games \cite{KMSTW:DGA}.
However, uniformity can be regained with $1$ bit of public memory,
i.e., there exist uniformly $\eps$-optimal Max strategies in countable
concurrent reachability games with finite action sets that are deterministic and use just $1$ bit of
public memory \cite{KMSTW:DGA}.
Optimal Max strategies in countable turn-based finitely branching
reachability games, where they exist, can be chosen to use just a step counter and $1$ bit
of public memory (but not just a step counter or just finite memory).
On the other hand, in concurrent games with finite action sets,
a step counter plus finite private memory does not suffice for optimal
Max strategies in general \cite{KMSTW:DGA}.

If Min is allowed infinite action sets (resp.\ infinite branching) 
in countably infinite reachability games, then Max generally needs infinite memory
for $\eps$-optimal, optimal and almost surely winning strategies.
There exists a turn-based countable reachability game with infinite Min
branching (and finite Max branching), such that every state admits an almost
surely winning strategy for Max,
and yet every Max strategy that uses only a step counter plus finite
private memory is still useless (in the sense that Min can make its
attainment arbitrarily close to zero) \cite{KMSTW:DGA}.

\bibliographystyle{splncs04}
\bibliography{journals,conferences,refs}

\end{document}